\documentclass[runningheads, envcountsame, a4paper]{llncs}

\usepackage{amssymb}
\usepackage[figuresright]{rotating}
\usepackage{subcaption}

\usepackage[utf8]{inputenc}
\usepackage{complexity}
\usepackage[algosection]{algorithm2e}
\usepackage{endnotes}
\usepackage{listings}
\usepackage{mathtools}
\usepackage{verbatim}
\usepackage{mathrsfs}
\usepackage{longtable}
\usepackage{enumitem}
\usepackage{etoolbox}
\usepackage{float}
\usepackage{color}
\usepackage{hyperref}

\usepackage{complexity}
\usepackage{comment}
\usepackage{longtable}

\usepackage{mathtools}

\newcommand*{\DEBUG}{}%

\ifdefined\DEBUG
\newcommand{\fixme}[1]{{\textcolor{red}{\bf{\textsf{FIXME: #1}}}}}
\newcommand{\bug}[1]{{\textcolor{blue}{\bf{\textsf{BUG: #1}}}}}
\newcommand{\idea}[1]{{\textcolor{blue}{\bf{\textsf{IDEA: #1}}}}}

\newcommand{\TODO}[1]{{\textcolor{red}{\bf{\textsf{ 
TODO: #1
}}}}}
\else
\newcommand{\fixme}[1]{}
\newcommand{\bug}[1]{}
\newcommand{\TODO}[1]{}
\newcommand{\idea}[1]{}
\fi

\newclass{\COMSLIP}{COM\mbox{-}SLIP}
\newclass{\COMSLIPCUP}{COM\mbox{-}SLIP^{\cup}}

\newclass{\DCM}{DCM}
\newclass{\eDCM}{eDCM}
\newclass{\eNPDA}{eNPDA}
\newclass{\DPDA}{DPDA}
\newclass{\RDPDA}{RDPDA}
\newclass{\PDA}{PDA}
\newclass{\DCMNE}{DCM_{NE}}
\newclass{\TwoDCM}{2DCM}
\newclass{\NCM}{NCM}
\newclass{\sNCM}{sNCM}
\newclass{\eNCM}{eNCM}
\newclass{\eNQA}{eNQA}
\newclass{\eNSA}{eNSA}
\newclass{\eNPCM}{eNPCM}
\newclass{\mNPCM}{mNPCM}
\newclass{\eNQCM}{eNQCM}
\newclass{\eNSCM}{eNSCM}
\newclass{\DPCM}{DPCM}
\newclass{\NPCM}{NPCM}
\newclass{\NQCM}{NQCM}
\newclass{\NSCM}{NSCM}
\newclass{\NPDA}{NPDA}
\newclass{\TRE}{TRE}
\newclass{\NFA}{NFA}
\newclass{\DFA}{DFA}
\newclass{\NCA}{NCA}
\newclass{\DCA}{DCA}
\newclass{\DTM}{DTM}
\newclass{\NTM}{NTM}
\newclass{\DLOG}{DLOG}
\newclass{\CFG}{CFG}
\newclass{\ULGC}{ULGC}
\newclass{\ULG}{ULG}
\newclass{\ETOL}{ET0L}
\newclass{\EDTOL}{EDT0L}
\newclass{\CFP}{CFP}
\newclass{\ORDER}{O}
\newclass{\MATRIX}{M}
\newclass{\BD}{BD}
\newclass{\LB}{LB}
\newclass{\ALL}{ALL}
\newclass{\decLBD}{decLBD}
\newclass{\StLB}{StLB}
\newclass{\SBD}{SBD}
\newclass{\TCA}{TCA}
\newclass{\RNCSA}{RNCSA}
\newclass{\RDCSA}{RDCSA}

\newclass{\DCSA}{DCSA}
\newclass{\NCSA}{NCSA}
\newclass{\DCSACM}{DCSACM}
\newclass{\NCSACM}{NCSACM}
\newclass{\NTMCM}{NTMCM}

\newclass{\SMG}{SMG}
\newclass{\RLSMG}{RLSMG}
\newclass{\LSMG}{LSMG}

\newclass{\NTPCM}{NTPCM}
\newclass{\NCSPCM}{NCSPCM}

\newclass{\PBCM}{PBCM}
\newclass{\NTPBCM}{NTPBCM}

\newclass{\NTCM}{NTCM}

\newclass{\UFIN}{\LL(IND_{UFIN})}
\newclass{\UFINONE}{\LL(IND_{{UFIN}_1})}
\newclass{\FIN}{\LL(IND_{FIN})}
\newclass{\ILIN}{\LL(IND_{LIN})}
\newclass{\ETOLfin}{\LL(ET0L_{FIN})}

\newclass{\PTIME}{PTIME}

\newclass{\GSM}{GSM}

\newsavebox{\spacebox}
\begin{lrbox}{\spacebox}
\verb*! !
\end{lrbox}
\newcommand{\blank}{\usebox{\spacebox}}%

\newcommand{\LL}{{\cal L}}
\newcommand{\MM}{{\cal M}}

\DeclareMathOperator{\yd}{yd}

\title{Techniques for Showing the Decidability of the Boundedness Problem of Language Acceptors
\thanks{The research of I.\ McQuillan was supported, in part, by Natural Sciences and Engineering Research Council of Canada Grant 2022-05092 (Ian McQuillan)
}}

\author{Oscar H. Ibarra\inst{1} \and Ian McQuillan\inst{2}}

\institute{Department of Computer Science\\ University of California, Santa Barbara, CA 93106, USA\\ \email{ibarra@cs.ucsb.edu}
\and Department of Computer Science, University of Saskatchewan\\
	Saskatoon, SK S7N 5A9, Canada\\
	\email{mcquillan@cs.usask.ca}
}

\authorrunning{O.H. Ibarra, and I. McQuillan }
\titlerunning{Containment Problem for Deterministic Multicounter Machine Models}

\tocauthor{Oscar~H.~Ibarra, Ian~McQuillan}

\begin{document}

\maketitle

\begin{abstract}
There are many types of automata and grammar models that have been studied in the literature, and for these models, it is common to determine whether certain problems are decidable.  One problem that has been difficult to answer throughout the history of automata and formal language theory is to decide whether a given system $M$ accepts a bounded language
(whether there exist words $w_1, \ldots,w_k$ such that $L(M) \subseteq w_1 \cdots w_k$?). Decidability of this problem has gone unanswered for the majority of automata/grammar models in the literature. 
Boundedness was only known to be decidable for regular and
context-free languages until recently when it was shown to also be decidable
for finite-automata and pushdown automata
augmented with reversal-bounded counters, and for vector addition systems with states. 

In this paper, we develop new techniques to show that the boundedness problem is decidable for larger classes of one-way nondeterministic automata and grammar models, by reducing the problem to the decidability of boundedness
for simpler classes of automata.
One technique involves characterizing the models
in terms of multi-tape automata.  We
give new characterizations of finite-turn Turing machines, 
finite-turn Turing machines augmented with various storage
structures (like a pushdown, multiple reversal-bounded counters,
partially-blind counters, etc.), and simple matrix grammars.
The characterizations are then used 
to show that the boundedness problem for these
models is decidable.
Another technique uses the concept of the store language of an automaton.
This is used to show that the boundedness problem is decidable for pushdown
automata that can ``flip'' their pushdown a bounded number of times, and 
boundedness remains decidable even if we augment this device with
additional stores.
\end{abstract}

\pagebreak

\section{Introduction}
\label{sec:intro}

There are many well-studied models of automata/grammars that are more powerful than
finite automata (denoted by $\NFA$) but less powerful than Turing machines. 
Perhaps the most
well-studied is the one-way nondeterministic pushdown automata ($\NPDA$) which accept the context-free languages. This model is
very practical --- for example, the {\em non-emptiness problem} (``given a machine $M$, is $L(M) \neq \emptyset$?''), as well as the {\em infiniteness problem} (``given a machine $M$, is $L(M)$ infinite?''), can both be determined in polynomial time for $\NPDA$ \cite{HU}.

Authors have studied models that are more powerful than $\NPDA$, such as $t$-flip $\NPDA$ (resp.\ finite-flip $\NPDA$), which have the ability to flip their pushdown stack at most $t$ (resp.\ a finite number of) times \cite{flipPushdown2,flipPushdown}. Non-emptiness and infiniteness are decidable for this model as well \cite{flipPushdown} (implied from their semilinear Parikh image).
Another more powerful model is simple matrix grammars, which are a class of grammars that generates a family of languages properly between the context-free languages and the matrix languages \cite{IbarraSMG}.

Other well-studied models with power between that of finite automata and Turing machines is the one-way nondeterministic reversal-bounded
multicounter machines \cite{Ibarra1978} ($\NCM$). This is an $\NFA$ with some number of counters, where each counter contains a non-negative integer, and transitions can detect if a counter is non-empty or not. The condition
of being $r$-reversal-bounded (resp.\ reversal-bounded) enforces that in each accepting computation, the
number of changes between non-decreasing and
non-increasing on each counter is at most $r$ (resp.\ a finite number).
It is also possible to combine different types of stores.
For example, another class of automata is $\NPDA$ augmented by
reversal-bounded counters, denoted by $\NPCM$.
This device, which is strictly more powerful than either $\NPDA$ or $\NCM$, 
has an $\NP$-complete non-emptiness problem \cite{HagueLin2011}. 
%This is quite a powerful model since combining two non-reversal-bounded counters, or combining two 1-turn pushdowns already has an undecidable emptiness problem.

We will also consider nondeterministic Turing machines with a one-way read-only input tape and a single two-way read/write worktape, denoted by $\NTM$. While all of the problems above are undecidable for $\NTM$, a {\em $t$-turn} (resp.\ {\em finite-turn}) $\NTM$ are machines with at most $t$ (resp.\ some number of) changes in direction on the worktape in every accepting computation (called reversal-bounded in \cite{visitautomata}, but we call it finite-turn here). 
%They also study this model where the
%worktape is pre-set to a string in a specified language $L$ from some given language family.
%**** IAN: IS THIS TRUE? IT DEPENDS ON THE PRESET LANGUAHE L, RIGHT? ****
Again, the non-emptiness and infiniteness problems are
decidable for this model. 
%Furthermore, finite-turn $\NTM$ has also been studied recently \cite{IbarraMcQuillanVerification} as it has desirable algorithmic properties for verification problems; the set of  successor (or predecessor) configurations of a given regular set of configurations is always a regular language. 

Another important property beyond emptiness and infiniteness is that of boundedness. A language $L \subseteq \Sigma^*$ is {\em bounded} if there exist non-empty words $w_1, \ldots, w_k$ such that $L\subseteq w_1^* \cdots w_k^*$. 
Here, we explore further the important decision problem called
the {\em boundedness problem}: given machine $M$, is $L(M)$ a
bounded language?''.
In the early years of the study of formal language theory, this property was shown to be decidable for $\NFA$ and $\NPDA$ by Ginsburg and Spanier using a rather complicated procedure \cite{ginspan,GinsburgCFLs}. In contrast, if a class of machines with an undecidable emptiness problem 
accepts languages that are closed under concatenation with the language $\$ \Sigma^*$ (where $\$$ is a new symbol and $\Sigma$ is an at least two letter alphabet),
then the boundedness problem is also undecidable for the class, because $\$\Sigma^*$ concatenated with anything non-empty is not
bounded, and so $L \$\Sigma^*$ is bounded if and only if $L$ is empty. Until recently, the status of the boundedness problem had been elusive for essentially all other machine/grammar models (besides $\NPDA$) studied in the literature that have a decidable emptiness problem. Finally, Czerwinski, Hofman, and Zetzsche showed that the boundedness problem is decidable for vector addition systems with states \cite{boundedVASS} (equivalent to one-way partially blind multicounter machines \cite{G78}, denoted by $\PBCM$, that properly contain $\NCM$). With $\PBCM$, machines can add and subtract from counters but cannot detect whether counters are empty or not except that a machine crashes if a counter goes below zero, and a word is accepted if it hits a final state with all counters being zero. Also, in \cite{Georg}, it was determined that the boundedness problem for $\NPCM$ (and $\NCM$) is not only decidable, but also $\coNP$-complete.

%%TO ADD BACK IN JOURNAL
\begin{comment}
Boundedness is intimately connected with decidability of the {\em containment problem} (``given $M_1, M_2$, is $L(M_1) \subseteq L(M_2)$?''). The containment problem is undecidable for even the simplest non-$\NFA$ class of machines: $\NFA$ with one
1-reversal-bounded counter \cite{Baker1974}.
Furthermore, Hopcroft showed that for a context-free language $L_0 \subseteq \Sigma^*$, the problem of deciding, ``given a context-free language $L \subseteq \Sigma^*$, is $L_0 \subseteq L$?'' is decidable if and only if $L_0$ is bounded \cite{boundedCFLDec}.  We call this the {\em bounded-containment decidability property}, and we study it here for other families besides the context-free languages.
\end{comment}
%

%The proof technique involves reducing the decidability of boundedness for $\NCM$ and
%$\NPCM$ to the decidability of boundedness for $\NFA$ and $\NPDA$.

Here, we develop techniques for showing that the
boundedness problem is decidable.
One technique involves creating characterizations in terms of multi-tape
versions of $\NFA$, $\NCM$, $\NPCM$, and $\PBCM$
of the following:
\begin{enumerate}
\item finite-turn $\NTM$ in terms of multi-tape $\NFA$,
\item finite-turn $\NTM$ augmented with reversal-bounded counters in terms of multi-tape $\NCM$,
\item finite-turn $\NTM$ augmented with a pushdown and reversal-bounded counters where in each accepting computation, the pushdown can only be changed 
during one sweep of the Turing tape, in terms of multi-tape $\NPCM$,
\item
finite-turn $\NTM$ augmented with partially blind counters in terms of multi-tape $\PBCM$.

\end{enumerate}
These characterizations are then used to show decidability of
the boundedness (also emptiness and infiniteness) problem for each of the models.
These results are strong as 
any combination of two 1-turn stores has an undecidable emptiness and thus boundedness problem.  In model (3) above, the restriction that the pushdown can only be used during one sweep (between two consecutive turns) of the read/write tape
cannot be dropped, as 
allowing one more sweep would make both emptiness and boundedness undecidable.
Note that the model in (3) is more powerful than $\NPDA$ and finite-turn $\NTM$, and can even accept non-indexed languages \cite{Aho}. For model (4), this model is strictly more powerful than 
the family of $\PBCM$ languages. Hence, it is the most powerful model containing non-semilinear languages with a known decidable boundedness problem. Using a similar technique, we show that the boundedness problem is decidable for simple matrix grammars (even when augmented with reversal-bounded counters).

Another technique involves the {\em store language} of a machine, which is the set of strings that encode the contents of the internal stores that can appear in any accepting computation.
There are some automata models in the literature where the family of store languages for that class can be accepted by a simpler type of automata. We use this to show that the boundedness problem
is decidable for finite-flip $\NPDA$. This is also true if augmented by reversal-bounded counters, and by a finite-turn worktape where the flip-pushdown is only used during one sweep of the worktape.
Hence, this is the most powerful model properly containing the context-free languages with a known decidable boundedness problem. 

All omitted proofs and some definitions are in the Appendix to help reviewers.

\section{Preliminaries and Notation}

We assume knowledge of introductory automata and formal language theory \cite{HU}, including deterministic and
nondeterministic finite automata, context-free grammars, pushdown automata, and Turing machines.

Let $\mathbb{N}$ be the set of positive integers and $\mathbb{N}_0$ the non-negative integers. Given a set $X$ and 
$t \in \mathbb{N}$, let $\langle X \rangle^t$ be the set of all $t$-tuples over $X$.
Given a finite alphabet $\Sigma$, let $\Sigma^*$ (resp.\ $\Sigma^+$) be the set of all words (resp.\ non-empty words) over $
\Sigma$. $\Sigma^*$ includes the empty word $\lambda$. A {\em language} $L$ is any subset of $\Sigma^*$, and a $t$-tuple
language $L$ is any subset of $\langle \Sigma^* \rangle^t$. Given a word $w$, the {\em reverse}
of $w$, denoted $w^R$ is equal to $\lambda$ if $w = \lambda$, and $a_n a_{n-1} \cdots a_1$ if $w = a_1 a_2 \cdots a_n, a_i \in \Sigma$ for $1 \le i \le n$. The
{\em length} of $w$, denoted by  $|w|$, is equal to the number of characters in $w$, and given $a\in \Sigma$, $|w|_a$ is the number of $a$'s in $w$.
Given alphabet $\Sigma = \{a_1, \ldots, a_m \}$ and $w \in \Sigma^*$, the {\em Parikh image} of $w$, $\psi(w) = (|w|_{a_1}, \ldots, |w|_{a_m})$; and the Parikh image of a language $L \subseteq \Sigma^*$ is $\psi(L) = \{\psi(w) \mid w\in L\}$.
Although we will not provide the formal definition of a language being semilinear, equivalently, a language is semilinear if and only if it has the same Parikh image as some regular language \cite{G78}.
Similarly, the Parikh image of $(w_1,\ldots, w_t) \in \langle \Sigma^* \rangle^t$, $\psi(w_1, \ldots, w_t) = \psi(w_1 \cdots w_t)$, and for 
$L \subseteq  \langle \Sigma^* \rangle^t, \psi(L) = \{\psi(x) \mid x \in L\}$.
A class of machines/grammars is said to be {\em effectively semilinear} if, given such a machine/grammar, a finite automaton with the same Parikh image can be effectively constructed.

\begin{comment}
A context-free grammar (abbreviated $\CFG$) is a tuple $G = (V, \Sigma, P, S )$, where $V$ and $\Sigma$ are finite and disjoint alphabets of nonterminals, and terminals, respectively, $S \in V$ is the starting nonterminal, and $P$ is a finite set of productions of the form $A \rightarrow w, A \in V, w \in (V \cup \Sigma)^*$. 
We will define language generation in terms of complete derivation trees \cite{HU}. A complete derivation tree is a directed
tree where vertices are labelled by an element of $V\cup\Sigma \cup \{\lambda\}$, the root is labelled by $S$, leaves are labelled by elements of $\Sigma \cup \{\lambda\}$, and if an inner-vertex is labelled by $A$ and its children are labelled by $u_1, \ldots, u_n$ from left-to-right, then $A \rightarrow u_1 \cdots u_n \in P$. The set of complete derivation trees is denoted by $T(G)$. The yield of a tree $t \in T(G)$, $\yd(t)$, is the concatenation of labels on leaves obtained by a preorder traversal. The language generated by $G$, $L(G) = \{\yd(t) \mid t \in  T(G)\}$.
 Any language generated by a context-free grammar is a context-free language. 
It is well-known that languages generated by $\CFG$s are equal to those accepted by $\NPDA$s \cite{HU}. 
\end{comment}

A $t$-tape $\NFA$ over $\Sigma$ is a generalization of an $\NFA$ where there are $t$ input tapes and they take $(w_1,\ldots, w_t) \in \langle \Sigma^*\rangle^t$ as input, and each transition is of the form $q' \in \delta(q,a,i)$, where the machine switches from state $q$ to $q'$ and reads $a \in \Sigma \cup \{\lambda\}$ from input tape $i$. This allows such a machine to accept a $t$-tuple language. The formal definition appears in the Appendix.

We will augment $t$-tape $\NFA$s with additional stores;
e.g.\ a $t$-tape $\NPDA$ is a $t$-tape $\NFA$ with an additional pushdown
alphabet $\Gamma$, and $\delta$ becomes a partial function with rules of the form $(q', \gamma) \in 
\delta(q,a,i, X)$, where $q,q',a,i$ are as with $t$-tape $\NFA$s, 
$X \in \Gamma$ is the topmost symbol of the pushdown which gets replaced by the word $\gamma \in \Gamma^*$.
Configurations now include a third component which contains the current pushdown contents, as is standard for pushdown automata \cite{HU}. We can similarly define machines with multiple stores by defining
the transitions to only read and change one store at a time. 
Standard one-way single-tape acceptors are a special case with only one input tape.
Multi-tape inputs 
%are somewhat less well studied than single tape variants, but 
have been studied for $\NPDA$ \cite{multitapeNPDA}, $\NCM$ \cite{Ibarra1978}, and $\NPCM$ \cite{IbarraMcQuillanVerification}.

We will use all machine models with a one-way read-only input, including those described in Section \ref{sec:intro}.
Almost all stores we consider are formally defined in \cite{StoreLanguages}, and we omit
the formal definitions due to space constraints (see Appendix). 
%The only other store considered here are partially-blind multicounter machines \cite{G78}.
%These are defined as $k$-counter machines where, from state $q$ and $a$ in the input alphabet or the empty word that uses
%counter $i$, the available transitions depending on whether counter $i$ is empty or not are the same. Unlike other counter models in this paper, language acceptance with partially blind machines is by final state and all counters being zero, which is important as they cannot otherwise detect that counters are zero \cite{G78}. However, configurations still only hold non-negative numbers and they can be thought of as ``crashing'' if they attempt to go below zero.

For all multi-tape machine models, they are effectively semilinear if and only if $1$-tape machines are effectively semilinear, because given a $t$-tape $M$, there is a $1$-tape machine $M'$ that reads $a$ from the tape whenever it can read $a$ from any tape of $M$; hence $\psi(L(M)) = \psi(L(M'))$.

It is known that for any type of nondeterministic machine model with reversal-bounded counters, one can equivalently use
monotonic counters \cite{IbarraGrammars} instead of reversal-bounded counters. 
Such machines have an even number $k$
of counters that we identify by $C_1, D_1, \ldots, C_{k/2}, D_{k/2}$ that can only be
incremented but not decremented, transitions do not detect the counter status, 
and acceptance occurs when the machine 
enters an final state with counters $C_i$ and $D_i$
having the same value for each $i$.  
Due to the equivalence, we will use the same notation as above
($\NPCM$, etc.) to mean machines with monotonic counters.
Monotonic counters are helpful in this paper because if we simulate 
an accepting computation of a machine with another machine that applies the same changes but in a different order, then the resulting simulation will still have matching monotonic counters.

\section{Boundedness Using Multi-Tape Characterizations}

\subsection{Characterizations of Finite-Turn Turing Machines}
\label{sec:finiteturn}

We first look at finite-turn $\NTM$, and finite-turn $\NTM$ with reversal-bounded counters (denoted by finite-turn $\NTCM$). These machines have previously been studied both without counters \cite{visitautomata}
and with counters \cite{Harju2002278}.
We give characterizations of these machines in terms of multi-tape $\NFA$ 
and  multi-tape $\NCM$.

\begin{comment}
\begin{example} \label{noncontextfree}
The language
$L = \{w\#w ~|~ w \in \{a,b\}^+ \}$ can be accepted by a 2-turn $\NTM$ that reads 
$w$ and puts it on the worktape, returns to the beginning of the tape and matches the worktape to the rest of the input. Note that $L$ is a classic example of
a non-context-free language \cite{HU}. \end{example}
\end{comment}

%\subsection{Normal Form and History Language}

\begin{example} \label{noncontextfreewithcounters}
Consider $L = \{w \# w \$ v \#v \mid w,v \in \{a,b\}^*, |w|_a = |v|_a, |w|_b = |v|_b\}$.
$L$ can be accepted by a $4$-turn $\NTCM$ $M$ with four monotonic counters as follows:
on input  $w_1 \# w_2 \$ v_1 \# v_2$, $M$ reads $w_1$ and writes it to the tape while in parallel recording $|w|_a$ and $|w|_b$ in two monotonic counters $C_1$ and $C_2$. When it hits $\#$, it turns and goes to the left end of the tape, and verifies $w_2 = w_1$. When it hits $\$$, it does the same procedure with the read/write tape to the right to verify $v_1 = v_2$, while in parallel putting $|v_1|_a$ and $|v_2|_b$ on two monotonic counters $D_1$ and $D_2$. It then accepts if the contents of $C_1$ equals $D_1$ and $C_2$ equals $D_2$. 
Although we do not have a proof, we conjecture this cannot be accepted by an $\NPCM$.
\end{example}

A $t$-turn $\NTM$ $M$ is in state normal form if $M$ makes exactly $t$ turns on all inputs accepted, the read/write head always moves to the right or left on every move, and in every accepting computation, it always turns to the left (resp.\ the right) on the same cell where it writes the current state on the tape. It can be shown (see Appendix) that any $t$-turn $\NTM$ can be transformed into another in state normal form that accepts the same language.
For such a $t$-turn $\NTM$ $M$ in state normal form,  define a $t+1$ tuple of symbols $b = (b_1, \ldots, b_{t+1})$ where each $b_i$ is in
the worktape alphabet $\Gamma$.
Let $\Delta$ be the alphabet of these symbols. We can think of a word in $\Delta^*$ as representing a $t+1$ track worktape,
where the $i$th component is the $i$th track.
For $1 \leq i \leq t+1$, define a homomorphism $h_i$ from $\Delta^*$ to $\Gamma^*$ such that $h_i((b_1, \ldots, b_{t+1})) = b_i$. 
Given a $t$-turn $\NTM$ $M$ in state normal form, define the {\em history language} $H(M)$ over $\Delta^*$ as follows:
$H(M)$ contains all strings $x$ where there is an accepting computation of $M$ such that $h_i(x)$ is the string
on the worktape after the $i$th sweep of the worktape and $h_{t+1}(k)$ is the string on the worktape at the end of the computation after it has made the final sweep after the last turn. This means if $t$ is even (it is similar if odd) 
\begin{equation*}h_1(x) = q_0 x_1 q_1, h_2(x) = q_2 x_2 q_1, \ldots , h_{t}(x) = q_t x_t q_{t-1}, h_{t+1}(x) = q_t x_{t+1} q_{t+1},\end{equation*} where $q_0$ is the initial state, $M$ writes $q_0 x_1 q_1$ on the first
sweep, etc. until $q_t x q_{t+1}$, which is the final worktape contents, and $q_{t+1}$ is a final state.

Let $t \ge 1$.  For a $t$-tuple $(w_1, \ldots, w_t)$, $w_i \in \Sigma^*$, let its {\em alternating pattern} be:
\begin{equation*}(w_1, \ldots, w_t)^A = \begin{cases}
w_1 w_2^R \cdots w_{t-1}w_t^R & \mbox{if $t$ is even,}\\
w_1 w_2^R \cdots w_{t-1}^Rw_t  & \mbox{if $t$ is odd.}\\
\end{cases}\end{equation*}
If there is a $t \ge 1$ with $L \subseteq \langle \Sigma^* \rangle^t$,
let $L^A = \{(w_1, \ldots, w_t)^A ~|~ (w_1, \ldots, w_t) \in L \}$.
We now show that every $t$-turn $\NTM$ can be ``converted'' to a
$(t+1)$-tape $\NFA$.
%Intuitively, the proof uses a normal form where the $\NTM$ moves right of left at ever step (no stay transitions), and always makes left (resp.\ right) turns on he same cell. 
Starting with $M$ in state normal form, $M'$ guesses a $(t+1)$-track string $x \in \Delta^*$ letter-by-letter from left-to-right 
while checking in parallel that the input on tape $i$ would be read by the simulated moves on track $i$ thereby verifying that $x \in H(M)$.
\begin{lemma} \label{NTMtoNFA}
Let $t \ge 0$, and $M$ be a $t$-turn $\NTM$  (resp.\ $t$-turn $\NTCM$). We can construct a $(t+1)$-tape $\NFA$  (resp.\ $(t+1)$-tape $\NCM$) $M'$ such that
$L(M' )^A =  L(M).$
\end{lemma}
For the opposite direction, on the first sweep of the worktape, $M'$  guesses and writes a guessed sequence of transition labels of $M$, and then sweeps the worktape once for each tape $i$ to make sure the next section of the input word of $M'$ would be read by tape $i$ in the simulation.
\begin{lemma} 
Let $t \geq 0$, and let $M$ be a $(t+1)$-tape $\NFA$ (resp.\ $(t+1)$-tape $\NCM$). Then we can construct a $t$-turn $\NTM$  (resp.\ $t$-turn $\NTCM$) $M'$ such that
$L(M' ) = L(M)^A.$
\label{NFAtoNTM}
\end{lemma}

From the two lemmas above, we obtain:
\begin{proposition} \label{char}
Let $t \geq 0$.
There is a $(t+1)$-tape $\NFA$ (resp.\ $(t+1)$-tape $\NCM$) $M$ if and only if there is a $t$-turn $\NTM$  (resp.\ $t$-turn $\NTCM$) $M'$ such that
$L(M' ) = L(M)^A$.
\end{proposition}

\begin{comment}

The following corollary
follows from Lemma \ref{NFAtoNTM} (since the $\NTM$
constructed in the proof is actually an $\NCSA$):

{\bf Check, I think Greibach also said this.

Actually, Greibach's result is slightly different. She
showed:  $L$ is accepted
by a $t$-turn $\NTM$ whose worktape is preset to strings in a regular
language $L$ iff $L$ is accepted by a $t$-turn $\NCA$ whose checking
stack is preset to strings in a regular language $L'$ ($L$ and $L'$ need 
not be the same.)}

\begin{corollary} \label{cor31}
A language $L$ is accepted by a $t$-turn $\NTM$ if and only if 
$L$ is accepted by a $t$-turn $\NCSA$.
\end{corollary}

\end{comment}

Further to the definition of bounded languages, we
say $L \subseteq \langle \Sigma^* \rangle^t$ is  a {\em bounded $t$-tuple language}
if $L \subseteq B_1 \times \cdots \times B_t$, where each $B_i$ is of the form 
$w_1^* \cdots w_n^*$ for some $w_1, \ldots, w_n \in \Sigma^+$. Given $L\subseteq \langle \Sigma^* \rangle^t$,
let $L^{(i)} = \{w_i \mid (w_1, \ldots, w_t) \in L\}$.

Let $\MM$ be a class of multi-tape machines consisting of an $\NFA$ with zero or more stores. Given a
$t$-tape $M \in \MM$, for each $i$, $1 \le i \le t$, let $M_i$ be the one tape machine in $\MM$ that simulates moves that
 read from tape $i$ by reading from the input tape, but reads $\lambda$ to simulate a read from other tapes.
So, $L(M_i) = L^{(i)}$ for all $i$, $1 \leq i \leq t$.
The following is easily verified:
\begin{lemma}
\label{ktapeto1}
A $t$-tape $M \in \MM$ is a bounded (resp.\ non-empty, finite) $t$-tuple language if and only if $L(M_i)$ is a bounded (resp.\ non-empty, finite) language for each  $1 \leq i \leq t$.
\end{lemma}
\begin{proof}
The proofs for non-emptiness and finiteness are clear. %Boundedness remains.

Assume $L(M)$ is a bounded $t$-tuple language, and therefore there exists $B_1, \ldots, B_t$, where each $B_i$ is of the form
$w_1^* \cdots w_n^*$ and $L(M) \subseteq B_1 \times \cdots \times B_t$. Thus, for each $i$, $L(M_i) \subseteq B_i$, and is bounded.

Assume each $L(M_i)$ is bounded, and let $B_i$ be such that
$L(M_{(i)}) \subseteq B_i$ and $B_i$ is of the form $w_1^* \cdots w_n^*$. Then, $L(M) \subseteq B_1 \times \cdots \times B_t$.
\qed \end{proof}

Using this characterization, we can show the following.
\begin{proposition} \label{boundedNTM}
The boundedness, non-emptiness, and infiniteness problems for finite-turn $\NTM$ (resp.\ finite-turn $\NTCM$) are decidable, and they are effectively semilinear.
\end{proposition}
\begin{proof}
From Proposition \ref{char}, given a $t$-turn $\NTM$ $M'$,
there is a $t+1$-tape $\NFA$ $M$ with $L(M)^A = L(M')$. We will decide if $L(M)^A$ is bounded; indeed, 
we will show $L(M)^A$ is bounded if and only if, for each $i$, $1 \le i \le t+1$, $L(M_i)$ is bounded.

Assume $L(M)^A = L(M')$ is bounded. Assume $t$ is odd (with the even case being similar). Then $L(M)^A =  \{w_1 w_2^R \cdots w_t w_{t+1}^R \mid (w_1, \ldots, w_{t+1}) \in L(M)\}$ is bounded. It is known that given any bounded language $L$, the reverse of $L$ is bounded, the set of subwords of $L$ is bounded, and any subset of $L$ is bounded \cite{GinsburgCFLs}. Hence, for each $i$, $\{w_i \mid (w_1, \ldots, w_{t+1}) \in L(M')\}$ is bounded as, for $i$ odd, then this is a subset of the set of subwords of $L(M)^A$, and for $i$ even, it is the reverse. This set is $L(M_i)$, and so each $L(M_i)$ is bounded.

Conversely, assume each $L(M_i)$ is bounded for $1 \le i \le t+1$, and hence $L(M_i)^R$ is also bounded for $i$ even. By Lemma \ref{ktapeto1}, $L(M)$ is a bounded $t+1$-tuple language. Since the finite concatenation of bounded languages is bounded \cite{GinsburgCFLs}, $L(M)^A$ is bounded.

Hence, $L(M')$ is bounded if and only if, for each $i$, $1 \le i \le t+1$, $L(M_i)$ is bounded. Since these are each regular, we can decide this property.

The proof is the same for finite-turn $\NTCM$ using decidability of boundedness for $\NCM$ \cite{Georg,boundedVASS}.

Semilinearity follows from semilinearity of $\NFA$ and $\NCM$ \cite{Ibarra1978}, and since $\psi(L(M)) = \psi(L(M'))$.
\qed \end{proof}
Although decidability for non-emptiness, finiteness, and effective semilinearity were known for both finite-turn $\NTM$ and finite-turn $\NTCM$ \cite{Harju2002278}, to our knowledge, decidability of boundedness for both finite-turn $\NTM$ and for finite-turn $\NTCM$ were not previously known.

\subsection{Finite-Turn $\NTM$ with Pushdown and Counters}
\label{sec:NTMwithPushdown}

In this section, we provide a further generalized model by augmenting
finite-turn $\NTM$ with not only monotonic counters but also a pushdown
where the pushdown can only be used in a restricted manner:

A $t$-turn $\NTM$ augmented with monotonic counters
and a pushdown is called a
$t$-turn $\NTPCM$. Such a machine is called {\em $i$-pd-restricted} if, during every accepting
computation, the pushdown is only
used  the $i$th sweep (either left-to-right or right-to-left) of the worktape; and a machine is {\em pd-restricted} if in every accepting
computation, the pushdown is only used in a single pass (it can be different passes depending on the computation).
%As mentioned in Section \ref{sec:intro}, without the restricted condition, emptiness and boundedness are undecidable.

\begin{example}
Let $D_1$ be the 
language over the alphabet $\{a_1,b_1\}$ generated by the context-free grammar with productions $S\rightarrow a_1 S b_1 S$ and $S \rightarrow \lambda$. This language is known as the
 Dyck language
over one set of parentheses, and let 
$L = \{x \# x \# x ~|~  x \in D_1 \}$.
$L$ can be accepted by a pd-restricted 4-turn $\NTPCM$ machine
(even without counters). Indeed, on input
$x_1 \# x_2 \# x_3$, a machine $M$ can use the pushdown to verify $x_1 \in D_1$ while in parallel copying $x_1$ to the read/write tape. Since this is the only pass where the pushdown is used, the machine is $1$-pd-restricted. Then it can match $x_1$ against the input to verify $x_1 = x_2=x_3$.
It follows from \cite{EngelfrietCopying} and \cite{Rozoy} that $L$ is not even an indexed language, a family that strictly contains the context-free languages \cite{Aho}, and is equal to the family of languages accepted by automata with a ``pushdown of pushdowns'' \cite{IteratedStack}. 
%This family is quite powerful while still having a decidable emptiness problem \cite{IteratedStack}. 
%It follows that $L$ is not an indexed language  by combining two results from the literature as follows: in \cite{EngelfrietCopying}, it is shown that for context-free language $L'$ that is not in the family $\EDTOL$ (which is a type of grammar system that we do not define here, but refer to \cite{EngelfrietCopying}), $\{x\# x \# x \mid x \in L'\}$ is not an indexed language; combined with $D_1$ not being in $\EDTOL$ \cite{Rozoy}. We also strongly believe $L$ cannot be accepted by an $\NPCM$ but do not yet have a proof. 
Thus, pd-restricted $\NTPCM$ is quite a powerful model, containing all of $\NPDA$, finite-turn $\NTM$, and even some non-indexed languages.% But we find it still possesses positive decidability properties.
\end{example}

The characterization will use a restriction of multi-tape $\NPCM$ as follows. Let $i$ satisfy $1 \le i \le t$.
A $t$-tape $\NPCM$ is {\em $i$-pd-restricted} if, for every accepting computation,
the pushdown is only used when it reads from a single input tape. 

The following is easy to verify. Any transition $\alpha$ that reads $a \in \Sigma \cup \{\lambda\}$ from input tape $j \neq i$ and uses the pushdown can be simulated by
first reading $a$ from tape $j$ but not changing another store, then in the next transition, reading and changing the pushdown as in $\alpha$ while reading $\lambda$ from tape $i$.
\begin{lemma} \label {rest}
Every $t$-tape $\NPCM$ can be converted to an equivalent
$i$-pd-restricted $t$-tape $\NPCM$, for any $1 \le i \le t$.  
\end{lemma}

The next proposition follows a proof similar to Lemma \ref{NTMtoNFA} for one direction, where the pushdown is only used in one track because it is only used on one sweep of the $\NTPCM$ Turing tape; and for the other direction it first uses Lemma \ref{rest} and then follows the proof of Lemma \ref{NFAtoNTM} where it is verified that the pushdown changes properly according to the guessed transition sequence by simulating the pushdown in the $i$th sweep of the Turing tape.
\begin{proposition} \label{pda1}
Let $t \ge 0$. There is a $(t+1)$-tape $\NPCM$ $M$
if and only if there is an $i$-pd-restricted $t$-turn $\NTPCM$ $M'$ 
such that $L(M' ) = L(M)^A$, for any $0 \le i \le t$.
\end{proposition}
\begin{comment}
\begin{proof}
The construction from $i$-pd-restricted $t$-turn $\NTPCM$ to  ($i$-pd-restricted) $(t+1)$-tape $\NPCM$ is similar to the
 proof of Lemma \ref{NTMtoNFA} noting
that since the pushdown in the $t$-turn $\NTPCM$ is only used 
between two consecutive turns (or the start or end of the computation) of the worktape head, the simulation of the pushdown is only done on a single track of the $(t+1)$-track guessed string when reading that one input tape.
(Note that the simulation of the pushdown is done in reverse if $i$ is odd).

For the reverse construction, by Lemma \ref{rest}, given a ($t+1$)-tape $\NPCM$, we can convert to an equivalent
pd-restricted ($t+1$)-tape $\NPCM$. To convert that to an
pd-restricted $t$-turn $\NTPCM$ ($0 \le i \le t$) is similar to the one in the
proof of Lemma \ref{NFAtoNTM} where it only verified that the pushdown changes properly 
according to the guessed transition sequence by simulating the pushdown only between two consecutive turns of the Turing tape (or the start or the end).
\qed \end{proof}
\end{comment}

For the next proof, we are able to strengthen the result to $M$ being pd-restricted rather than $i$-pd-restricted, seen as follows:
Given a pd-restricted machine $M$, we can make $M_1, \ldots, M_{t+1}$, where each $M_i$ accepts the strings accepted by $M$
for which the pushdown is used in sweep $i$. Thus, $M_i$ is $i$-pd-restricted. Furthermore, $L(M)$ is bounded if and only if
$L(M_i)$ is bounded for each $i$.
The remaining proof is similar to that of
Proposition \ref{boundedNTM} using the fact that emptiness, infiniteness, and
boundedness for $\NPCM$ are decidable \cite{Ibarra1978,Georg}.

\begin{proposition} \label{NTPCM}
For every $t \ge 0$, the boundedness, emptiness, and infiniteness problems for pd-restricted $t$-turn $\NTPCM$ are
decidable, and they are effectively semilinear.
\end{proposition}

%The above proposition can be generalized.
Briefly, the result above can be generalized by replacing the pushdown with other potential types of stores.
Consider an $\NFA$ augmented with a storage structure $S$ and the specification for updating $S$ and possibly some necessary  condition(s) on $S$ for
acceptance, in addition to the machine entering a final state. The storage structure $S$ can include multiple storage structures.  We do not
define such a storage structure formally for simplicity, and the following result can be thought of as a template for other models where decidability of boundedness can be shown. 
However, definitions such as storage structures \cite{EngelfrietCheckingStack} and store types \cite{StoreLanguages} work.
Examples of $S$ are:
pushdown; reversal-bounded counters (or equivalent storage structures such as
monotonic counters);
partially blind counters; and combinations of the structures, e.g., a pushdown and reversal-bounded
 counters.  
 
 We can examine $S$-restricted $t$-turn $\NTM$ augmented by $S$ (denoted
 $\NTM(S)$) where in every accepting computation, $S$ can only be changed within a single sweep of the worktape.
We can show the following seeing that the pushdown in the proof above can be replaced with other storage types. 
\begin{proposition} \label{G3}
Let $\MM$ be a class of $\NFA$ with storage structure $S$, whose
languages are closed under reversal, where the boundedness (resp.\ emptiness, infiniteness)
problem for $\MM$ are decidable.  Then for every $t \ge 0$,
the boundedness (resp.\ emptiness, infiniteness) problem for
$S$-restricted $t$-turn $\NTM(S)$
 are decidable. Furthermore, if $\MM$ is augmented with additional reversal-bounded counters (no restrictions on their use) has a decidable boundedness (resp.\ emptiness, infiniteness) problem, the corresponding problem for $S$-restricted $t$-turn machines with reversal-bounded counters are decidable.
\end{proposition}

This provides new results for certain general types of automata with decidable properties. For example, checking stack automata provide a worktape that can be written to before the first turn, and then only operate in read-only
mode. They have a decidable emptiness and infiniteness problem. If we augment these with a finite-turn worktape
where the checking stack could only be used in a single sweep, emptiness and infiniteness are decidable. 

As with other results in this paper, a $t$-turn $\NTM$ (combined with other stores)
can be replaced with a $t$-turn checking stack.
The restriction on $\NTPCM$ to be $S$-restricted in Proposition \ref{NTPCM}
is needed, as the next proposition shows. Let $\DCSA$ be deterministic
checking stack automata.
The first point
uses undecidability of non-emptiness of the intersection of two 1-turn deterministic pushdown automata \cite{Baker1974},
the second problem uses undecidability of the halting problem for Turing machines \cite{HU},
the third point uses the undecidability of the halting
problem for 2-counter machines \cite{Minsky}, and the fourth point uses the third point.

\begin{proposition} \label{undec}
The emptiness (boundedness, infiniteness) problems are undecidable
for the following models:
\begin{enumerate}
\item
1-turn $\NTM$ (or $\DCSA$) with a 1-turn pushdown. 
\item
2-turn $\DCSA$ with a 1-turn pushdown, even when the pushdown is used
only during the checking stack reading phase (i.e., after turn 1).
\item 1-turn $\NTM$ with an unrestricted counter.
\item 1-turn deterministic pushdown automata with an unrestricted counter.

\end{enumerate}
\end{proposition}

\subsection{Finite-Turn $\NTM$ with Partially Blind Counters}

Partially blind counter machines are multicounter machines ($\PBCM$) where the counters can be incremented or decremented but not
tested for zero, however the machine crashes if any of the 
counters becomes negative, and acceptance occurs when the machine enters an
accepting state with all the counters being zero. The emptiness, finiteness, \cite{G78} and boundedness problems \cite{Georg} have
been shown decidable for vector addition systems with states, which are equivalent to partially blind multicounter machines \cite{G78}.  The family of languages accepted by these machines
is a
recursive family, does not contain all context-free languages (as in the example below), but contains non-semilinear languages \cite{G78} (unlike all the other models considered so far in this paper).

Here, we look at $t$-turn $\NTM$ augmented with partially blind counters,
called $t$-turn $\NTPBCM$. 
%We also call such a machine $i$-restricted (for some $0 \le i \le t$) if during every accepting computation, the counters are only changed during the $i$th sweep of the Turing tape. Acceptance occurs at the end of the input if it hits a final state with all counters being $0$.
%\begin{example}
It is pointed out in \cite{G78} that
$L = \{w \#w^R ~|~ w \in \{a,b\}^*\}$ is not accepted by any $\PBCM$. However, it is easily accepted by a $\NTPBCM$ (or even a 2-turn $\NTM$), which are therefore 
strictly more powerful.
%\end{example}

We could augment a $\NTPBCM$ (or $\PBCM$, $t$-tape $\PBCM$) with
monotonic counters, but it is straightforward to see that each pair of monotonic counters can be simulated by a pair of
partially-blind counters, and we therefore do not consider these machines with additional reversal-bounded counters.

\begin{comment}
these counters do not add any additional power
to the machine, as any machine with $2k$ monotonic counters can be simulated by partially-blind
counters as follows: Given any $\NTPBCM$ ($\PBCM$, $t$-tape $\PBCM$) $M$ with monotonic
counters, we can build an equivalent $\NTPBCM$
($\PBCM$, $t$-tape $\PBCM$) $M'$.
Let $C_1,D_1, \ldots, C_k,D_k$ be the monotonic counters of $M$.
Then $M'$ simulates $M$.  If $M$ enters an accepting state, then
for each $1 \le i \le k$, $M'$ decrements $C_i$ and $D_i$
simultaneously a nondeterministically guessed number of times until it
guesses that they are both zero (if they were equal before
the decrement). Then $M'$ enters an accepting state.
\end{comment}

The results in Section \ref{sec:finiteturn} concerned finite-turn $\NTM$, optionally augmented with monotonic counters.  We
will see next that these results hold if ``monotonic counters'' is replaced
by ``partially blind counters''.  However, monotonic counters are easy to handle, as we can permute the order that counter changes are applied in an accepting computation and the resulting computation does not change the counter values. But this is not so for partially blind counters as changing the orders can cause counters to go below zero, which is not allowed. But we 
can modify the proof as follows:
\begin{proposition} \label{multi-PBCM}
Let $t \ge 0$. There is a $(t+1)$-tape $\PBCM$ $M$
if and only if there is a $t$-turn $\NTPBCM$ $M'$ 
such that $L(M' ) = L(M)^A$.
\end{proposition}
\begin{proof}
One half of the proof follows an identical construction to that in the proof of
Lemma \ref{NFAtoNTM} where all counters are simulated on the first sweep while guessing the transition sequence and the order of counter changes is the same.

For the reverse direction, the construction is a
modification to that of Lemma \ref{NTMtoNFA}. We describe
the construction of $M$ from $M'$.
Note that $M'$ makes $s = (t+1)$ left-to-right and right-to-left
sweeps on its worktape. If $M'$ has $k$ partially blind
counters $C_1, \ldots, C_k$, $M$ will have $sk$ partially blind counters called
$C_{1,j}, \ldots, C_{k,j}$ for $1 \le j \le s$.
The  simulation  of all sweeps of the computation of $M'$
on its worktape are done in parallel.
The counters in $C_{1,j}, \ldots, C_{k,j}$ are used to simulate the counters
of $M$ in sweep $j$. For odd $j$, the simulation is faithful, but
for even $j$, the simulation is backwards. The counters in $C_{i,1}$ are initially zero, as are $C_{i,s}$ if $s$ is even. For all $j$ even, $C_{i,j}$ and 
$C_{i,j+1}$ are set nondeterministically to be the same guessed values. The simulation of the computation of $M'$
on the even tracks of the worktape (using counters
$C_{i,j}$, $j$ even) is done in reverse and in parallel with
the simulation of the odd tracks (using counters in $C_{i,j}$
$j$ odd).  When the simulation reaches the end of the worktape and
$M'$ enters an accepting state, for all odd $j$, $j<s$, the counters in $C_{i,j}$ and $C_{i,j+1}$
are decremented simultaneously a nondeterministically guessed number
of times to verify  that they are the same (and if $j=s$ is is not changed thereby verifying that it is zero).
Then  $M$ enters a final state. Figure \ref{pblindfigure} demonstrates an example. This technique allows counters to be adjusted in a different order.
\qed \end{proof}

\begin{figure}[t]
\centering
\subcaptionbox{}{\includegraphics[width=.45\textwidth]{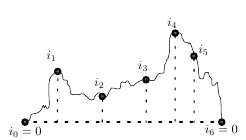}}
\subcaptionbox{}{\includegraphics[width=.45\textwidth]{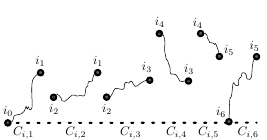}}
\caption{In (a), it shows the counter values of counter $i$ is an original accepting computation of $M'$, where $i_j$ is the value at turn $j$ of the Turing tape. In (b), we see the modified computation of $M$ with each counter simulating in parallel.}
\label{pblindfigure}
\end{figure}

From Proposition \ref{multi-PBCM} and decidability of boundedness for $\PBCM$, we obtain:
\begin{proposition} \label{boundPB}
The boundedness, emptiness, and finitenesss problems for finite-turn $\NTPBCM$ are decidable.
\end{proposition} 
We believe that this is a new result for all three decision problems, and in particular, decidability of boundedness
is quite powerful.

\subsection{Simple Matrix Grammars}

%We are going to show that boundedness is decidable for simple matrix grammars \cite{IbarraSMG} which
%is a special type of matrix grammars.
A matrix grammar
%First, a {\em matrix grammar of degree $n$} is a tuple $G = (V,\Sigma,P,S)$ where $V$ is the set of
%nonterminals, with $S \in V$, $\Sigma$ is the disjoint set of terminals, and $P$ is 
has a finite set of matrix rules of the form
$[ A_1 \rightarrow w_1, \ldots, A_k \rightarrow w_k ]$,
where each $A_i \rightarrow w_i$ is a context-free production.
%For $\alpha, \beta \in (V \cup \Sigma)^*$, we let $\alpha \Rightarrow \beta$ if
%$\alpha = \alpha_0, \beta = \alpha_k, \alpha_i = u_i A_i v_i, u_i,v_i \in (V\cup \Sigma)^*, \alpha_{i+1} = u_i w_i v_i$ for $0 \le i \le k$,
%and $[A_1 \rightarrow w_1, \ldots, A_k \rightarrow w_k] \in P$. We let $\Rightarrow^*$ be the reflexive and transitive closure of $\Rightarrow$, and the language generated by $G$, $L(G) = \{w \in \Sigma^* \mid S \Rightarrow^* w\}$.
In the derivation, at each step a matrix is chosen nondeterministically, whereby the context-free rules of the matrix must be applied in order to the sentential form to produce the next sentential form.
An $n$-simple matrix grammar ($n$-$\SMG$) (from \cite{IbarraSMG}), a restricted form
of a matrix grammar, is a tuple $G = (V_1, \ldots , V_n, \Sigma, P, S)$, where
$V_1, \ldots , V_n$ are disjoint sets of nonterminals, $\Sigma$
is the terminal alphabet, $S$ is a start nonterminal not in
$(V_1 \cup \cdots \cup V_n)$, and $P$ is a finite set of rules of the form:
\begin{enumerate}
\item $S \rightarrow A_1 \cdots A_n$, where each $A_i \in V_i$,
\item $[X_1 \rightarrow w_1, \ldots , X_n \rightarrow w_n]$, where $X_i \in V_i$ and $w_i \in (V_i \cup \Sigma)^*$, and the number of nonterminals in
$w_i$ is equal to the number of nonterminals in $w_j$ for all $i \ne j$.
\end{enumerate}
The derivation relation enforces that in each rule of type 2., always the leftmost nonterminal of $V_i$ in the sentential form is rewritten (precise definition of the derivation relation is in \cite{IbarraSMG}).
The language $L(G)$ consists of all strings $w \in \Sigma^*$
that can be derived starting from $S$ and applying the rules in such a
leftmost derivation.  
Note that a $1$-$\SMG$ is just a context-free grammar. It is known \cite{IbarraSMG} that
without either the restriction of the number of nonterminals being the same, or not requiring leftmost derivation,
grammars can generate more languages.

The following result follows from \cite{IbarraSMG} and \cite{multitapeNPDA}. 
\begin{proposition} \label{smnocounter}
$L$ is generated by an $n$-$\SMG$ $G$  if and only if there is an $n$-tape $\NPDA$
$M$ accepting $L' \subseteq \langle \Sigma^* \rangle^n$ such that $L = \{ x_1 \cdots x_n  \mid (x_1, \ldots , x_n) \in L'\}$.
\end{proposition}

We can generalize the definition of a simple matrix grammar by augmenting
it with monotonic counters. Then, in every matrix rule, each context-free production includes
$2k$ counter increments, and for a string to be generated, the counter values
in counter $i$ and $i+1$ have to be equal, for $i$ odd.
Proposition \ref{smnocounter} can be generalized to include counters.
\begin{proposition}
$L$ is generated by an $n$-$\SMG$ with $2k$ monotonic counters $G$ if
and only if there a an $n$-tape $\NPDA$ with $2k$ monotonic counters (which
is equivalent to an $n$-tape $\NPCM$) $M$ accepting $L' \subseteq \langle \Sigma^* \rangle^n$ such
that $L = \{ x_1 \cdots x_n \mid  (x_1, \ldots , x_n) \in L'\}$.
\end{proposition}
\begin{comment}
\begin{proof}
Suppose $L'$ is accepted
by an $n$-tape $\NPDA$ $M$ with input alphabet $\Sigma$ and $2k$ monotonic counters
$C_1,D_1, \ldots, C_k, D_k$.  Let $a_1, b_1, \ldots, a_k, b_k$ be new symbols.  We
construct an $n$-tape $\NPDA$ $M'$ with input alphabet
$\Delta = \Sigma \cup \{a_1,b_1,\ldots , a_k,b_k\}$.  $M'$ when given an input $w \in \Delta^*$,
simulates $M$ on $w$ but in a move, instead of incrementing  $C_r$ ($D_r$) by a
non-negative integer  $i_r$ ($j_r$), $M'$ reads $a_r^{i_r}$ ($b_r^{j_r}$) on input tape $1$.
By Proposition \ref{smnocounter}, since $M'$ is a $n$-tape $\NPDA$, we can construct an $n$-$\SMG$ $G'$
generating $L(M')$.  Next, we construct an $n$-$\SMG$ with $2k$ monotonic counters $G$ from $G'$ as follows:
If in a rule $R'$ of $G'$, the symbol $a_r$ ($b_r$) appears  $i_r$ ($j_r$) times, we create
a rule $R$ of $G$ by deleting  these  symbols and adding  $i_r$ ($j_r$) as increment
to counter $C_r$ ($D_r$).  Clearly, $L(G) = \{x_1 \cdots x_n \mid (x_1, \ldots, x_n) \in L\}$.

The converse is proved by reversing the construction above:  Given an
$n$-$\SMG$ with $2k$ monotonic counters $G$, we construct an $n$-$\SMG$ $G'$
with $2k$ new terminal symbols to simulate the increments $i_r$ ($j_r$)
to counter $C_r$ ($D_r$) by generating $a_r^{i_r}$ ($b_r^{j_r}$) on the
first components of the rules.  From $G'$ we then construct an $n$-tape $\NPDA$
$M'$.  Finally from $M$, we construct a $n$-tape $\NPDA$ with monotonic counters $M$.
\qed \end{proof}
\end{comment}

Using a proof similar to the decidability problems shown using other multi-tape characterizations in this paper, we obtain:
\begin{proposition} The emptiness, infiniteness, and boundedness problems
for simple matrix grammars (resp.\ with monotonic counters) are decidable, and they are effectively semilinear.
\end{proposition}
\begin{comment}
\begin{proof}
Given an $n$-$\SMG$ with $2k$ monotonic counters, we construct an $n$-tape $\NPDA$ $M$ with monotonic
counters (which is equivalent to a $n$-tape $\NPCM$) such that
$L (G) = \{x_1 \cdots x_n  \mid (x_1, \ldots, x_n) \in L(M)\}$.  
Let $L_i = \{x_i  \mid  (x_1, \ldots, x_n) \in L(M)\}$.
We construct for each $i$, an $\NPCM$ $M_i$ accepting $L(M_i)$.  Then $L(G)$ is non-empty  (resp.\ finite, bounded) if and only if each $L(M_i)$ is non-empty (resp.\ finite, bounded).
The result follows since these problems are decidable for $\NPCM$. Semilinearity is also clear.
\qed \end{proof}
\end{comment}

\section{Store Languages for the Boundedness Problem}

To summarize, so far we have determined several new classes of machines for which the boundedness problem is
decidable. One of the largest is finite-turn
$\NTM$ with reversal-bounded counters and a pushdown where, in each accepting computation, the pushdown can only be used within a single sweep of the Turing worktape. In this section, we determine one more class that is even more general than this one.
The algorithm provides an entirely different technique than multi-tape characterizations that we have used thus far.

We focus on finite-flip $\NPDA$ \cite{flipPushdown}.
A $t$-flip (resp.\ finite-flip) $\NPCM$ augments a $t$-flip $\NPDA$ with $k$-reversal-bounded counters.
With this model, configurations are of the form
$(q,w,Z_0 \gamma, i_1, \ldots, i_k)$ where $q$ is the current state, $w$ is the remaining input, $Z_0 \gamma$ is the current pushdown contents, and $i_j$ is the current contents of counter $j$.

In \cite{StoreLanguages,IbarraMcQuillanVerification}, the authors study the concept of a {\em store language of a machine $M$} for arbitrary types of automata, which is essentially a language description of all the store contents that can appear in any accepting computation of the machine.
So, for a $t$-flip $\NPCM$ $M = (Q,\Sigma, \delta,q_0,F)$, the store language of $M$, 
\begin{equation*}S(M) = \{ q Z_0 \gamma c_1^{i_1} \cdots c_k^{i_k} \mid \begin{array}[t]{l}  (q_0, w, Z_0, 0, \ldots, 0) \vdash_M^* (q, w', Z_0\gamma, i_1, \ldots, i_k) \vdash_M^* \\  ( q_f, \lambda, Z_0 \gamma', i_1', \ldots, i_k'),  \\ q_f \in F, w,w' \in \Sigma^*, \gamma' \in \Gamma^*, i_1', \ldots, i_k'  \ge 0\} ,\end{array}\end{equation*}
where $c_1, \ldots, c_k$ are new special symbols associated with the counters.
In \cite{IbarraMcQuillanVerification}, the authors showed that the store language of every $t$-flip $\NPDA$ (resp.\ $t$-flip $\NPCM$) is in fact a regular (resp.\ $\NCM$) language. Therefore, the pushdown can be essentially eliminated. This is a generalization of the  important result that the store language of any $\NPDA$ is regular \cite{GreibachCFStore}.

The next proof uses an inductive procedure (informally described without counters) where we know $0$-flip $\NPDA$s (equal to the context-free languages) have a boundedness problem. And inductively, if we have an $r+1$-flip $\NPDA$, we can create two machines, an $r$-flip machine that accepts the parts of the inputs of $M$ read during the first $r$ flips that eventually leads to acceptance, and a $0$-flip machine that accepts the parts of the inputs of $M$ from which, with no flips, it will eventually accept. These two languages use the store languages, which can be accepted by finite automata.
The purpose of the store languages should be noted. Simply using the fact that $r+1$-flip $\NPDA$ are
closed under gsm mappings, it is immediately evident that both of these languages can be accepted by $r+1$-flip
$\NPDA$ (just using closure properties). But, by using the store language, it is possible to accept the first with only
an $r$-flip $\NPDA$ and the second with a $0$-flip $\NPDA$. This is needed to make the induction work, so that essentially we can decide boundedness up to any given $r$.

\begin{proposition}
The boundedness, emptiness, and infiniteness problems are decidable for finite-flip $\NPDA$ (resp.\ finite-flip $\NPCM$).
\end{proposition}

Lastly, we consider machines with a finite-flip pushdown, reversal-bounded counters, and a finite-turn worktape. Such a machine is {\em pd-restricted} if, in every accepting computation, the finite-flip pushdown can only be used in one left-to-right sweep or right-to-left sweep of the worktape.
Finally, by Proposition \ref{G3}:
\begin{corollary}
The class of pd-restricted finite-flip $\NPCM$ augmented with a finite-turn worktape has a decidable boundedness, emptiness, and infiniteness problem.
\end{corollary}

\section{Conclusions}
In this paper, we study powerful one-way nondeterministic machine models, and find new models where the boundedness, emptiness, and infiniteness problems are decidable. The largest of these are finite-turn Turing machines augmented by partially blind counters, and finite-turn Turing machines augmented by a pushdown that can be flipped a finite number of times, and reversal-bounded counters, where the pushdown can only be used in one sweep of the Turing worktape. It also shows two new techniques to show these problems are decidable.

\vfil\eject

\vfil\eject

\section*{Appendix}

\subsection*{Definitions}

For $t \ge 1$, a {\em one-way $t$-tape nondeterministic finite automaton} ($t$-tape $\NFA$) is a tuple $M = (Q,\Sigma,\delta,q_0,F)$
where $Q$ is a finite set of states, $\Sigma$ is the finite input alphabet, $q_0 \in Q$ is the initial state, 
$F\subseteq Q$ is the set of final states, and $\delta$ is a partial function from
$Q \times (\Sigma \cup \{\lambda\}) \times \{i \mid 1 \le i \le t \}$ (for 1-tape machines, we unambiguously leave off the last component) to finite subsets of $Q$. We usually denote an element 
$q' \in \delta(q,a,i)$ by $\delta(q,a,i) \rightarrow q'$. 
%We often associate labels from an alphabet $T$ bijectively to the transitions. 
%We then define a homomorphism `${\rm read}$' from 
%$T^*$ to $\{1, \ldots, t\}^*$ such that ${\rm read}(\alpha) = i$ when transition $\alpha \in T$  reads from input tape $i$.
A {\em configuration} of $M$ is a tuple $(q,(w_1, \ldots, w_t))$ where $q \in Q$ is the current state, 
and $(w_1, \ldots, w_t), w_1, \ldots, w_t \in \Sigma^*$  is the remainder of the $t$-tape input. Two configurations change as follows:
\begin{equation*}(q,(w_1, \ldots, w_{i-1}, aw_i, w_{i+1},\ldots, w_t)) \vdash (q',(w_1, \ldots, w_t)),\end{equation*}
if there is a transition $\delta(q,a,i) \rightarrow q'$. 
%We also write $\vdash$ to represent $\vdash^{\alpha}$
%for some transition $\alpha$. We extend $\vdash^{\alpha}$ to words $\vdash^{\alpha_1 \cdots \alpha_n}$
%where $\alpha_1 \cdots \alpha_n \in T^*$ in the
%natural way to represent derivations of length zero or more.
We let $\vdash^*$ be the reflexive and transitive closure of $\vdash$.
An {\em accepting computation} on $(w_1, \ldots , w_t) \in \langle \Sigma^* \rangle^t$ is a sequence
\begin{equation} (q_0, (w_1, \ldots, w_t)) \vdash \cdots \vdash (q_n, (\lambda,\ldots,\lambda)), \label{computation} \end{equation} 
where $q_n \in F$. 
%The string ${\rm read}(\alpha_1 \cdots \alpha_n)$ is called the {\em tape order} of the computation (\ref{computation}).
The {\em language accepted} by $M$, $L(M) \subseteq \langle \Sigma^* \rangle^t$ is the set of all $(w_1,\ldots,w_t)$ for which there is an accepting computation. 
%The {\em tape order language} $T(M)$ is the set of all tape orders of every 
%accepting computation; so $T(M) \subseteq \{1,\ldots, t\}^*$.

A $t$-tape
$k$-counter machine is a tuple $M = (Q,\Sigma,\delta,q_0,F)$, where $Q,\Sigma,q_0,F$ are just like $t$-tape $\NFA$,
and $\delta$ has transitions
$\delta(q,a,i,s,j) \rightarrow (q',x)$ where $q,q' \in Q, a\in \Sigma \cup\{\lambda\}, 1 \le i \le t,  s \in \{0,1\}, x \in \{-1,0,1\}, 1 \le j \le k$.
A configuration is a tuple $(q,(w_1,\ldots,w_t), (z_1, \ldots, z_k))$ where $q \in Q$ is the current state, $(w_1, \ldots,w_t)$ is the remaining contents of the input tapes, and $(z_1, \ldots, z_k)$ are the contents of the counters, where $z_i \in \mathbb{N}_0$, for each $i$. Configurations change by 
\begin{eqnarray*}&&(q,(w_1, \ldots, w_{i-1}, aw_i, w_{i+1},\ldots, w_t), (z_1,\ldots, z_k)) \vdash^{\alpha} \\ && \hphantom{yyyyyyyyyyyyyyyyyyyy} (q',(w_1, \ldots, w_t), (z_1, \ldots, z_{j-1}, z_j+x, z_{j+1}, \ldots, z_k)),\end{eqnarray*} if $\alpha$ is $\delta(q,a,i,s,j) \rightarrow (q',x)$, $s$ is $0$ if $z_j = 0$, and $s$ is 1 if $z_j$ is positive.
As such, $s$ is known as the counter status as it is used to check if a counter is empty or not. A $k$-counter machine $M$ is $r$-reversal-bounded (resp.\ reversal-bounded) if in each accepting computation, the
number of changes between non-decreasing and
non-increasing (or vice versa) on each counter is at most $r$ (resp.\ a finite number).
A $k$-counter machine is {\em partially-blind} if $\delta(q,a,i,0,j)=  \delta(q,a,i,1,j)$ for each $q \in Q, a \in \Sigma \cup \{\lambda\},
1 \le i \le t, 1 \le j \le k$. For this reason, typically the counter status component is left off the transitions.

We also examine one-way machines with a two-way (Turing) read/write worktape denoted by $\NTM$. 
These machines (for this model, we only use 1-tape inputs) have the same components as 
1-tape $\NFA$, but also have a worktape alphabet $\Gamma$, (including a fixed blank character $\blank$), 
and $\delta$ is 
from  $Q \times (\Sigma \cup \{\lambda\}) \times \Gamma$ to subsets of
$Q \times \Gamma \times \{{\rm L},{\rm S}, {\rm R}\}$. Each transition $\delta(q,a,y) \rightarrow (q',z,x)$ consists of, the current state $q \in Q$, the state to switch to $q' \in Q$, the input $a \in \Sigma \cup \{\lambda\}$,  the symbol currently being scanned on the worktape $y \in \Gamma$, the symbol to replace it with $z \in \Gamma$, and the direction $x$ (left, stay, or right) moved by the read/write head. 

An $\NTM$ (resp.\ $\NPDA$)  $M$ is {\em $l$-turn} if, in every accepting computation, the worktape makes at most $l$ changes in
direction, between moving towards
the right and moving towards the left, and vice versa. 
A machine is {\em finite-turn} if it is $l$-turn for some $l$.

Let $M$ be a $t$-turn $\NTM$ (resp.\ $t$-turn $\NTCM$), where $t \ge 0$. 
We say $M$ is in {\em normal form} if:
%\begin{enumerate}
$M$ makes exactly $t$ turns on all inputs accepted;
the read/write worktape head always moves left or right at every step that uses the worktape (no stay transitions);
on every accepting computation, there is a worktape cell $d$, and $M$ only turns left on cell $d$ and right on cell $1$ (the cell it starts on);
the worktape never moves left of cell $1$ or right of cell $d$; and
$M$ accepts only in cell $1$ or $d$.
%\end{enumerate}

%The following is straightforward to show.
\medskip
\noindent {\bf Normal Form Lemma.}
%\label{statenormalform}
Let $t \geq 0$. Given a $t$-turn $\NTM$ (resp.\ $t$-turn $\NTCM$) $M$, we can construct
a $t$-turn $\NTM$  (resp.\ $t$-turn $\NTCM$) $M'$ in normal form such that $L(M') = L(M)$.
\begin{proof}
First, we  assume without loss of generality that $M$ starts by moving towards the right (if it does not,
then another machine can be built which uses the worktape in the opposite direction).

Next, we introduce three new worktape symbols:  $\rhd, \#$, and $\lhd$ (plus a marker that can be added to any letter). $M'$ starts by writing $\rhd$ on the first cell. Throughout, if $M$ has counters, then transitions that use them are simulated verbatim. Then $M'$ simulates $M$ before the first turn, whereby at each step, $M'$ can either simulate a transition verbatim that moves right. It can simulate a sequence of stay transitions before moving right as follows: if the symbol on the worktape cell is $x$, it guesses the final contents of the cell before eventually moving right, $y$, and replaces $x$ with $y$, but by immediately moving right. It then simulates the sequence of stay transitions appropriately from $x$ to $y$ but using the state to store the current simulated symbol and by moving right on $\#$ on each step, ultimately verifying that the simulated sequence of stay transitions ends with $y$. Also, at each step, instead of simulating a transition of $M$, it can instead nondeterministically write any number of $\#$ symbols and move right on the store (reading $\lambda$ on the input). At some point, $M'$
writes a $\lhd$ on the worktape (this is the guessed rightmost cell $d$ to be visited during the entire computation). Next, it continues the simulation but only towards the left by only simulating
transitions that move left verbatim (only on a non-$\#$ symbol from the worktape), or that stay on the worktape in a similar fashion as above starting on $x$, guessing the final contents $y$, replacing $x$ with $y$, then simulating the sequence appropriately using the state, by moving left on $\#$ at each step. It can also skip over arbitrarily many $\#$ symbols on $\lambda$ input. If it simulates a turn transition, it instead marks the current cell, moves left to $\rhd$ and back to the marked cell, where it unmarks it and continues the simulation. If $t\ge 2$, it again switches direction and continues
this same simulation towards the right in a similar fashion, and so on.
$M'$ remembers how many turns have occurred in the finite control to make sure it turns exactly $t$ times. Lastly, if $M'$
hits a final state of $M$, it remembers this in the finite control and continues making a full $t+1$ sweeps of the 
worktape whence it enters a final state of $M'$.

Altogether, $M'$ can shuffle in arbitrarily many $\#$ symbols into the worktape, it writes $\rhd$ on the first cell,
$\lhd$ on some cell $d$ (nondeterministically guessed so that $M$ would not move to the right of that cell on any turn), it only turns at those two designated cells, and always turns at those two cells, there are no
stay transitions, and it accepts only on either $\rhd$ or $\lhd$. It is evident that $L(M') = L(M)$ and $M'$ is in normal form.
\qed \end{proof}

Also, given such a machine $M$ in normal form, we can have the
machine write the current state in the first and last cell (1 and $d$) every time it reaches them. We call this
{\em state normal form}.

\setcounter{theorem}{1}
\begin{lemma} 
Let $t \ge 0$, and let $M$ be a $t$-turn $\NTM$  (resp.\ $t$-turn $\NTCM$). We can construct a $(t+1)$-tape $\NFA$  (resp.\ $(t+1)$-tape $\NCM$) $M'$ such that
$L(M' )^A =  L(M).$
\end{lemma}
\begin{proof}
We will only describe the case when $k$ is odd, with the even case being similar.

Assume without loss of generality by the previous lemma that $M$ is in state normal form.
We construct $M'$ which accepts input $(w_1, \ldots, w_{t+1})$ if and only if $w_1 w_2^R \cdots w_t w_{t+1}^R \in L(M)$, as follows:

On input $(w_1, \ldots, w_{t+1})$, $M'$ guesses a $(t+1)$-track string $x \in \Delta^*$ letter-by-letter (here, $x$ does not need to be stored as it is guessed one letter at a time from left-to-right), and simulates the computation of $M$ on the $t+1$ input
tapes by making sure that the computation of $M$ is ``compatible'' with the guessed string $x$ while checking that $x \in H (M)$. 
To do this, $M'$ verifies that on input $w_1$, $M$ could read $w_1$ before the first turn and finish its first sweep with $h_1(x)$ on its tape, on input $x_2^R$ and starting with state and tape contents of $h_1(x)$, $M$ could finish its second sweep with $h_2(x)$ on its tape, etc. Furthermore, on the guessed $x$, it will be possible to verify in parallel that, for each $i$, $1 \le i \le t$, $h_i(x)$ produced $h_{i+1}(x)$ while reading $w_i$ if $i$ is odd, or $w_i^R$ if $i$ is even. Also, it is possible to do so from left-to-right when guessing $x$, even when $i$ is even. We will describe the even case which is slightly more complicated. The 
simulation using the even tracks are done in reverse by ``flipping the directions''. It only needs to simulate transitions that move left. To simulate $\delta(q,a,y) \rightarrow (p, z, {\rm L})$ on track $i$, $M'$ switches the simulation of track $i$ from $p$ to $q$ while reading $a$ from tape $i$, and verifying that the $(t+1)$-track string $x$ has $y$ in the current letter of track $i-1$ (or $\blank$ if $i=1$), and $z$ in the current letter of track $i$.
If $M$ has monotonic  counters, then $M'$ simulates them verbatim, as applying transitions in a different order preserves their total.
\qed \end{proof}

\begin{lemma} 
Let $t \geq 0$, and let $M$ be a $(t+1)$-tape $\NFA$ (resp.\ $(t+1)$-tape $\NCM$). Then we can construct a $t$-turn $\NTM$  (resp.\ $t$-turn $\NTCM$) $M'$ such that
$L(M' ) = L(M)^A.$
\end{lemma}
\begin{proof}
We will describe the case when $t$ is odd, with the even case being similar. First consider the case without counters.

Each transition of the $(t+1)$-tape $\NFA$ is of the form
$\delta(q, a,i) \rightarrow p$,
%   $\delta(q,a_1, \ldots, a_{k+1}) \rightarrow
where $q,p$ are states, $a \in \Sigma \cup \{\lambda\}$ and $1 \le i \le t+1$.
Let $T$ be a set of labels in bijective correspondence with the transitions of $M$. 

On input $w = w_1 w_2^R  \cdots w_t w_{t+1}^R$ (note that the
$w_i$'s need not have the same lengths and can even be the empty word since some transitions are
on $\lambda$ input), $M'$ operates as follows:
\begin{enumerate}
\item $M'$ uses the worktape to guess a sequence of transition labels of $M$. So $M'$ writes $\alpha_0 \alpha_1 \cdots \alpha_n$ (each $\alpha_i \in T$),
where it verifies that $\alpha_0$ is a transition from an initial state of $M$, $\alpha_n$ is a transition into a final state, and the ending state of $\alpha_i$ is the starting state of $\alpha_{i+1}$ for all $i$, $0 \leq i < n$. 
In parallel, $M'$ reads input $w_1$ and verifies that the letters of $\Sigma \cup \{\lambda\}$ on input tape 1 that are read by the transition sequence $\alpha_0 \cdots \alpha_n$ are $w_1$.
\item $M'$ turns on the worktape and reads input $w_2^R$ and makes sure the letters read by $\alpha_n \cdots \alpha_0$ on tape 2 are $w_2^R$.
\item $M$ turns on the worktape and reads $w_3$ and makes sure the letters read by $\alpha_0 \cdots \alpha_n$ on tape 3 are $w_3$.
\item[$\vdots$]
\item[] until tape $t+1$.
\end{enumerate}

Because the transition sequence $\alpha_0 \cdots \alpha_n$ is fixed after step 1, $M'$ can verify that $(w_1, \ldots, w_{t+1})$ could
be read by $\alpha_0 \cdots \alpha_n$ by making a turn on the store after reading each $w_i$.

If $M$ has counters, then the guessed sequence of transition labels remains the same, and all counter changes can be applied when guessing it.
\qed \end{proof}

\setcounter{theorem}{8}
\begin{proposition}
Let $t \ge 0$. There is a $(t+1)$-tape $\NPCM$ $M$
if and only if there is an $i$-restricted $t$-turn $\NTPCM$ $M'$ 
such that $L(M' ) = L(M)^A$, for any $0 \le i \le t$.
\end{proposition}
\begin{proof}
The construction from $i$-restricted $t$-turn $\NTPCM$ to  ($i$-restricted) $(t+1)$-tape $\NPCM$ is similar to the
 proof of Lemma \ref{NTMtoNFA} noting
that since the pushdown in the $t$-turn $\NTPCM$ is only used 
between two consecutive turns (or the start or end of the computation) of the worktape head, the simulation of the pushdown is only done on a single track of the $(t+1)$-track guessed string when reading that one input tape.
(Note that the simulation of the pushdown is done in reverse if $i$ is odd).

For the reverse construction, by Lemma \ref{rest}, given a ($t+1$)-tape $\NPCM$, we can convert to an equivalent
$i$-restricted ($t+1$)-tape $\NPCM$. To convert that to an
$i$-restricted $t$-turn $\NTPCM$ ($0 \le i \le t$) is similar to the one in the
proof of Lemma \ref{NFAtoNTM} where it only verified that the pushdown changes properly 
according to the guessed transition sequence by simulating the pushdown only between the $i$th and $i+1$st turns of the Turing tape (or the start or the end).
\qed \end{proof}

\setcounter{theorem}{11}
\begin{proposition}
The emptiness (boundedness, infiniteness) problems are undecidable
for the following models:
\begin{enumerate}
\item
1-turn $\NTM$ (or $\DCSA$) with a 1-turn pushdown. 
\item
2-turn $\DCSA$ with a 1-turn pushdown, even when the pushdown is used
only during the checking stack reading phase (i.e., after turn 1).
\item 1-turn $\NTM$ with an unrestricted counter.
\item 1-turn deterministic pushdown automata with an unrestricted counter.

\end{enumerate}
\end{proposition}
\begin{proof}
It was pointed out in the introduction that undecidability of emptiness implies undecidability of boundedness. 
Similarly, given $M$, $L(M)\Sigma^*$ is infinite if and only if $L(M) \neq \emptyset$, and hence undecidability
of emptiness implies undecidability of infiniteness.

The first item above follows from the undecidability of emptiness
of  the intersection of languages
accepted by 1-turn $\DPDA$ (deterministic $\NPDA$) \cite{Baker1974}.

For the second item, we will use the undecidability of the halting
problem for single-tape $\DTM$ on an initially blank tape.
Let $Z$ be a single-tape $\DTM$. Define the following language:

$L = \{I_1 \# I_3 \cdots  \# I_{2k-1} \$ I_{2k}^R \# \cdots \# I_4^R \# I_2^R ~|~
I_1 \Rightarrow \cdots \Rightarrow I_{2k-1} \Rightarrow I_{2k}$ is a halting computation  of $Z \}$.

Construct a 2-turn $\DCSA$ with a 1-turn pushdown $M$ as follows when given 
an input of the form
$w = I_1 \# I_3 \cdots \#I_{2k-1} \$ I_{2k}^R \# \cdots \# I_4^R \# I_2^R$:

\begin{enumerate}
\item
$M'$ writes $I_1 \# I_3 \cdots \#I_{2k-1}$ on the checking stack.
\item
$M'$ turns on the checking stack, and while reading input
$I_{2k}^R \# \cdots \# I_4^R \# I_2^R$ does the following in parallel:

-- It checks that $I_1 \Rightarrow I_2, I_3 \Rightarrow I_4, \ldots, I_{2k-1} \Rightarrow I_{2k}$.

-- It pushes $I_{2k}^R \# \cdots \# I_4^R \# I_2^R$  on the pushdown.
\item

$M'$ then makes a second turn on the checking stack and checks
  (by popping the pushdown and scanning the checking stack) that 
   $I_2 \Rightarrow I_3, \ldots, I_{2k-2} \Rightarrow I_{2k-1}$.
\end{enumerate}
$M$ makes only 2 turns on the checking stack and 1-turn on 
the pushdown, and $L(M)$ is empty if and only if $Z$ does not halt.

For the third point, we show that it is undecidable whether a machine of this type accepts $\lambda$.
The proof of the undecidability of the halting problem for 2-counter machine
with counters $C_1$ and $C_2$  (with no input tape) in \cite{Minsky}  shows that
the counters operate in phases.  A phase begins with one counter,  say
$C_1$,  having
value  $d_i$ and the other counter, $C_2$,  having value $0$.  During the phase, $C_1$
decreases while $C_2$ increases. The phase ends with $C_1$ having value 0 and
$C_2$ having value $e_i$. Then in the next phase the modes of the counters
are interchanged: $C_2$ decreases to zero while $C_1$ increases to $d_{i+1}$.
At the start, $d_1 = 1$. Thus, a halting computation of $M$ (if it halts)
will be of the form:

\begin{center}
$(q_1, d_1,  0) \Rightarrow^* (q_2, 0, e_1) \Rightarrow^* (q_3, d_2, 0) \Rightarrow^* (q_4, 0, e_2)
\Rightarrow^* \cdots \Rightarrow^*(q_{2k}, 0, e_k)$
\end{center}

\begin{comment}
a sequence of configurations
corresponding to the phases will  be of the form
$(q_1, d_1,  0), (q_2, 0, d_2),
(q_3, d_3, 0), (q_4, 0, d_4), ....$
\end{comment}

\noindent
where $q_1, \ldots, q_{2k}$ are states and $d_1, e_1, d_2, e_2, ..., d_{2k}, e_{2k}$
are positive integers with $d_1 = 1$ and the shown configurations are the ends of the phases. Note that the second component of the configuration
refers to the value of $C_1$, while the third component refers to the value of $C_2$.  We assume that if $M$ halts, it halts with zero in counter $C_1$.

We construct a 1-turn $\NTM$ $M'$ with an unrestricted counter $D$ to simulate the
2-counter machine $M$ on $\lambda$ input as follows:
\begin{enumerate}
\item
$M'$ writes $z = a^{d_k}\# b^{d_k}\# \cdots \# a^{d_3}\# b^{d_3}\#  a^{d_2}\# b^{d_2}\# a^{d_1}\# b^{d_1}$
on its read/write tape, where  $d_1 = 1$ and $k, d_k, \ldots , d_3, d_2$ are nondeterministically chosen  positive integers.
Clearly, $M'$ can do this without reversing on the read/write tape
with the help of $D$; When $M'$ writes $a^{d_i}$, it simultaneously
increments $D$ by $d_i$, and then it decrements $D$ to zero while
writing $\#b^{d_i}$.
\item
$M'$ then reverses its read/write head and simulates $M$. In the
simulation, the 1-turn read/write tape will keep track of 
the changes in counter $C_1$ and $D$ will simulate $C_2$.
The simulation is done as follows:  Suppose counter $C_1$
has value $d_i$ represented by $a^{d_i}$ and
the read/write head is on $\#$ to the left of $b^{d_i}$, and
counter $C_2$ is zero. $M'$ moves the
read/write head left to $\#$ simulating $M$ and incrementing 
$D$ to $e_i$.  This simulates the phase where $C_1$ which has
value $d_i$ is decremented to zero while $C_2$ is incremented
to $e_i$.  In the next phase, $M'$ simulating $M$ decrements $D$
to zero while moving the read/write head left of $b^{d_{i+1}}$
to the next $\#$ checking that $d_{i+1}$ is valid. 
At this point, $D$ is zero, and counter $C_1$ has value 
represented by $a^{d_{i+1}}$.  The process is then repeated.
\end{enumerate}
Clearly, $M'$ accepts $\lambda$ if and only if $M$ halts,
which is undecidable.

Given a 1-turn $\NPDA$ $M$ with an unrestricted counter
operating on $\lambda$ input,  let $T = \{t_1, \ldots, t_k\}$
be its set of transitions.  Construct a 1-turn $\DPDA$  $M'$
with an unrestricted counter with inputs in $T^+$ which
operates as follows when given input  $w  = a_1 \cdots a_n$
in $T^+$:  $M'$ checks that the
transition $a_1$ is applicable to the initial condition, i.e.,
the state is $q_0$, the stack symbol $Z_0$, and counter 0.
Then $M'$ tries to simulate $M$ guided by the
transitions in $w$.  $M'$ accepts $w$ if and only if 
the sequence of transitions $w$ leads to
$M$ to accept $\lambda$.  It follows that $L(M')$ is not empty
if and only if $M$ accepts $\lambda$.  The result follows
from the third point of this proposition.

\qed \end{proof}

\setcounter{theorem}{15}

\begin{proposition}
$L$ is generated by an $n$-$\SMG$ with $2k$ monotonic counters $G$ if
and only if there a an $n$-tape $\NPDA$ with $2k$ monotonic counters (which
is equivalent to an $n$-tape $\NPCM$) $M$ accepting $L' \subseteq \langle \Sigma^* \rangle^n$ such
that $L = \{ x_1 \cdots x_n \mid  (x_1, \ldots , x_n) \in L'\}$.
\end{proposition}
\begin{proof}
Suppose $L'$ is accepted
by an $n$-tape $\NPDA$ $M$ with input alphabet $\Sigma$ and $2k$ monotonic counters
$C_1,D_1, \ldots, C_k, D_k$.  Let $a_1, b_1, \ldots, a_k, b_k$ be new symbols.  We
construct an $n$-tape $\NPDA$ $M'$ with input alphabet
$\Delta = \Sigma \cup \{a_1,b_1,\ldots , a_k,b_k\}$.  $M'$ when given an input $w \in \Delta^*$,
simulates $M$ on $w$ but in a move, instead of incrementing  $C_r$ ($D_r$) by a
non-negative integer  $i_r$ ($j_r$), $M'$ reads $a_r^{i_r}$ ($b_r^{j_r}$) on input tape $1$.
By Proposition \ref{smnocounter}, since $M'$ is a $n$-tape $\NPDA$, we can construct an $n$-$\SMG$ $G'$
generating $L(M')$.  Next, we construct an $n$-$\SMG$ with $2k$ monotonic counters $G$ from $G'$ as follows:
If in a rule $R'$ of $G'$, the symbol $a_r$ ($b_r$) appears  $i_r$ ($j_r$) times, we create
a rule $R$ of $G$ by deleting  these  symbols and adding  $i_r$ ($j_r$) as increment
to counter $C_r$ ($D_r$).  Clearly, $L(G) = \{x_1 \cdots x_n \mid (x_1, \ldots, x_n) \in L\}$.

The converse is proved by reversing the construction above:  Given an
$n$-$\SMG$ with $2k$ monotonic counters $G$, we construct an $n$-$\SMG$ $G'$
with $2k$ new terminal symbols to simulate the increments $i_r$ ($j_r$)
to counter $C_r$ ($D_r$) by generating $a_r^{i_r}$ ($b_r^{j_r}$) on the
first components of the rules.  From $G'$ we then construct an $n$-tape $\NPDA$
$M'$.  Finally from $M$, we construct a $n$-tape $\NPDA$ with monotonic counters $M$.
\qed \end{proof}

\begin{proposition} The emptiness, infiniteness, and boundedness problems
for simple matrix grammars (resp.\ with monotonic counters) are decidable, and they are effectively semilinear.
\end{proposition}
\begin{proof}
Given an $n$-$\SMG$ with $2k$ monotonic counters, we construct an $n$-tape $\NPDA$ $M$ with monotonic
counters (which is equivalent to a $n$-tape $\NPCM$) such that
$L (G) = \{x_1 \cdots x_n  \mid (x_1, \ldots, x_n) \in L(M)\}$.  
Let $L_i = \{x_i  \mid  (x_1, \ldots, x_n) \in L(M)\}$.
We construct for each $i$, an $\NPCM$ $M_i$ accepting $L(M_i)$.  Then $L(G)$ is non-empty  (resp.\ finite, bounded) if and only if each $L(M_i)$ is non-empty (resp.\ finite, bounded).
The result follows since these problems are decidable for $\NPCM$. Semilinearity is also clear.
\qed \end{proof}

\setcounter{theorem}{17}

\begin{proposition}
The boundedness, emptiness, and infiniteness problems are decidable for finite-flip $\NPDA$ (resp.\ finite-flip $\NPCM$).
\end{proposition}
\begin{proof}
It suffices to prove it with counters.
We will prove by induction on $t \ge 0$, that every $t$-flip $\NPCM$ has a decidable boundedness problem. The base case when 
$t=0$ is true because every $0$-flip $\NPCM$ is in fact a normal $\NPCM$ which has a decidable boundedness problem \cite{Georg}.

Let $r \geq 0$, assume that every $r$-flip $\NPCM$ has a decidable boundedness problem, and let $M = (Q,\Sigma,\delta,q_0, F)$ be a $(r+1)$-flip $\NPCM$.
Assume without loss of generality that every flip transition of $M$ is on $\lambda$, and that 
every flip transition that can be used for the $i$th flip is from states in $P_i$ to $P_i'$ where these states are not used for any other
transitions that do not involve the $i$th flip.
Also assume that each state implies the topmost symbol of the stack (i.e.\ each transition guesses the  topmost stack symbol, and then in the next step verifies that it guessed correctly). Lastly, assume without loss of generality that in every accepting computation, $M$ makes exactly $r+1$ flips.
Then $L(M)$ equals
\begin{eqnarray*}
&& \{ w v \mid \begin{array}[t]{l} (q_0, w, Z_0, 0, \ldots, 0) \vdash^* (q_1, \lambda, Z_0 \gamma, i_1, \ldots, i_k) \mbox{~with~} r \mbox{~flips}, q_1 \in P_{r+1}, \\
(q_1,\lambda, Z_0\gamma, i_1, \ldots, i_k) \vdash (q_2,\lambda, Z_0 \gamma^R, i_1, \ldots, i_k) \mbox{~with one flip}, q_2 \in P_{r+1}' \\
\mbox{and~} (q_2, v,Z_0\gamma^R, i_1, \ldots, i_k) \vdash^* (q_3, \lambda, \gamma', i_1', \ldots, i_k') \mbox{~with no flips}, q_3 \in F \}  \end{array}
\end{eqnarray*}
Let $X$ be the set of all pairs $(w,v)$ in $L(M)$ above.

Consider $S_1 = S(M) \cap P_{r+1} \Gamma^*c_1^* \cdots c_k^*$ and $S_2 = S(M) \cap P_{r+1}' \Gamma^*c_1^* \cdots c_k^*$.  Because the store language of
every finite-flip $\NPCM$ is an $\NCM$ language \cite{IbarraMcQuillanVerification} and $\NCM$ is closed under intersection with regular languages, $S_1$ an $S_2$ are in $\NCM$.
Because also $\NCM$ is closed under reversal, we can build an $\NCM$ $M_1$ that accepts the reversal of $S_1$, and $M_2$ can be built that accepts $S_2$ (not the reversal). Let $k_1, k_2$ be the number of counters in $M_1$ and $M_2$ respectively.

Let $L_1$ equal to
\begin{eqnarray*}
& = & \{w \mid \exists v, (w, v)  \in  X\}\\
&= & \{ w \mid \begin{array}[t]{l} \exists q_1,q_2 \in Q, q_3 \in F, \gamma, \gamma' \in \Gamma^*, v \in \Sigma^*, i_j, i_j' \ge 0 \mbox{~for~}1 \le j \le k \mbox{~such that~} \\ (q_0, w, Z_0,0, \ldots, 0) \vdash^* (q_1, \lambda, Z_0 \gamma, i_1, \ldots, i_k) \mbox{~with~} r \mbox{~flips}, q_1 \in P_{r+1},\\
(q_1,\lambda, Z_0\gamma, i_1, \ldots, i_k) \vdash (q_2,\lambda, Z_0 \gamma^R, i_1, \ldots, i_k) \mbox{~with one flip}, q_2 \in P_{r+1}'\\
\mbox{and~} (q_2, v,Z_0\gamma^R, i_1,\ldots, i_k) \vdash^* (q_3, \lambda, Z_0 \gamma', i_1', \ldots, i_k') \mbox{~with no flips} \} \end{array}\\
&=& \{w \mid \begin{array}[t]{l} (q_0,w,Z_0, 0, \ldots, 0) \vdash^* (q_1,\lambda, Z_0\gamma, i_1, \ldots, i_k) \mbox{~with~} r \mbox{~flips},\\ \mbox{and~} q_1 Z_0 \gamma c_1^{i_1} \cdots c_k^{i_k} \in S_1\}.\end{array}
\end{eqnarray*}
This is a $r$-flip language because a $r$-flip $\NPCM$ with $k+k_1$ counters can be built that simulates $M$ until an arbitrarily guessed state $q \in Q_1$ (using the first $k$ counters),
where, if it has $Z_0 \gamma$ on the pushdown and $i_1,\ldots, i_k$ in the counters, it simulates $M_1$ (on the other $k_1$ counters) on $(q Z_0\gamma c_1^{i_1} \cdots c_k^{i_k})^R$ by reducing each counter to zero from counter $k$ to counter $1$, then popping from the pushdown until it is empty.

Let $L_2$ equal to
\begin{eqnarray*}
& = & \{v \mid \exists w, (w, v)  \in  X\}\\
&= & \{ v \mid \begin{array}[t]{l} \exists q_1,q_2 \in Q, q_3 \in F,\gamma,\gamma' \in \Gamma^*,w \in \Sigma^*, i_j, i_j'  \ge 0 \mbox{~for~} 1 \le j \le k \mbox{~such that} \\
(q_0, w, Z_0, 0, \ldots, 0) \vdash^* (q_1, \lambda, Z_0 \gamma, i_1, \ldots, i_k) \mbox{~with~} r \mbox{~flips}, q_1 \in P_{r+1},\\
(q_1,\lambda, Z_0\gamma, i_1, \ldots, i_k) \vdash (q_2,\lambda, Z_0 \gamma^R, i_1, \ldots, i_k) \mbox{~with one flip}, q_2 \in P_{r+1}'\\ 
\mbox{~and~} (q_2, v,Z_0\gamma^R, i_1, \ldots, i_k) \vdash^* (q_3, \lambda, Z_0 \gamma', i_1',\ldots, i_k') \mbox{~with no flips} \} \end{array}\\
&=& \{v \mid \begin{array}[t]{l}(q_2,v,Z_0\gamma, i_1 \ldots, i_k) \vdash^* (q_3,\lambda, Z_0\gamma', i_1, \ldots, i_k') \mbox{~with~} 0 \mbox{~flips,~} q_3 \in F, \\ i_1', \ldots, i_k' \ge 0, \mbox{~and~} q_2 Z_0 \gamma c_1^{i_1} \cdots \gamma c_k^{i_k} \in S_2\}.\end{array}
\end{eqnarray*}
This is a $\NPCM$ language (no flips) as a machine with $k+ k_2$ counters can be built which 
simulates $M_2$ by guessing and checking that $q_0 Z_0 \gamma c_1^{i_1} \cdots c_k^{i_k} \in S_2$, 
while pushing $Z_0 \gamma$ onto the pushdown and adding $i_j$ to each counter $j$. Then it uses the other $k$
counters to simulate $M$ from $q$ without any flips.

\begin{claim}
$L$ is bounded (resp.\ non-empty, finite) if and only if $L_1$ and $L_2$ are both bounded (resp.\ non-empty, finite).
\end{claim}
\begin{proof}
Assume $L$ is bounded. Then there exists $w_1, \ldots, w_n$ such that $L\subseteq w_1^* \cdots w_n^*$. But $L_1$ and $L_2$ are both subsets of subwords of $L$ and so they are bounded (the set of subwords of a bounded language is bounded, and any subset of a bounded language is bounded \cite{GinsburgCFLs}).

Assume $L_1$ and $L_2$ are bounded. Hence, there exists $w_1, \ldots, w_n, v_1, \ldots, v_m$ such that $L_1 \subseteq w_1^* \cdots w_n^*$
and $L_2 \subseteq v_1^* \cdots v_m^*$. It is immediate that $L \subseteq L_1 L_2$ because $L = \{xy \mid (x,y) \in X\}$. Hence, $L\subseteq w_1^* \cdots w_n^* v_1^* \cdots v_m^*$.
\qed \end{proof}

As we have an algorithm that checks if a ($0$-flip) $\NPCM$ is bounded \cite{Georg}, the proof above creates an algorithm to check if a $1$-flip $\NPDA$ is bounded. This provides an inductive algorithm that works up to an arbitrary number of flips.
\qed \end{proof}

\end{document}